\documentclass{sig-alternate-05-2015}

\mathcode`l="8000
\begingroup
\makeatletter
\lccode`\~=`\l
\DeclareMathSymbol{\lsb@l}{\mathalpha}{letters}{`l}
\lowercase{\gdef~{\ifnum\the\mathgroup=\m@ne \ell \else \lsb@l \fi}}%
\endgroup

\usepackage[utf8]{inputenc}
\usepackage[T1]{fontenc}   
\usepackage{lmodern}

\def\gathen#1{{#1}}

\usepackage[shortlabels]{enumitem}
\setlist[description]{font=\normalfont\itshape,itemsep=0ex,partopsep=0ex}

\usepackage{microtype} 
\usepackage{mathtools}  
\usepackage{booktabs}   
\usepackage{array}      
\usepackage[nocompress,nosort]{cite}

\usepackage{float}
\usepackage[noend]{algpseudocode}
\newfloat{algo}{tb}{lop}
\floatname{algo}{Algorithm}
\usepackage{caption}  
\usepackage{amssymb,amsmath}
\usepackage{amsfonts}
\usepackage{xcolor}
\usepackage{url}
\usepackage[colorlinks=true,citecolor=blue,hypertexnames=false]{hyperref}
\usepackage{bm}

\newtheorem{thm}{Theorem}
\newtheorem{definition}[thm]{Definition}
\newtheorem{prop}[thm]{Proposition}
\newtheorem{lem}[thm]{Lemma}

\newtheorem{paradigma}{Example}
\newenvironment{example}[1][]{\begin{paradigma}[#1]\normalfont\begin{small}}{\hfill \end{small}\end{paradigma}}

\newenvironment{algoenv}[3][\linewidth]{
\begin{minipage}{#1}%
\flushleft
\rule{\textwidth}{.08em}\vspace{-\baselineskip}\smallskip
\begin{description}[noitemsep]
\item[\rlap{Input}\phantom{Output}] #2
\item[Output] #3
\end{description}
\vspace{-\baselineskip}
\rule{\textwidth}{.05em}
\begin{algorithmic}
}{\end{algorithmic}
\vspace{-.5\baselineskip}
\rule{\textwidth}{.08em}
\end{minipage}}

\newcommand{\softO}{\tilde{{O}}}
\newcommand{\bigO}{{{O}}}

\newcommand{\Id}{\operatorname{I}}

\def\MhlEqOrder{{K}}

\def\bF{\mathbb{F}}

\def\bZ{\mathbb{Z}}

\title{Fast Computation\\ of  the Nth Term of an Algebraic Series\\ over a Finite Prime Field\titlenote{\small %
We warmly thank the referees for their very helpful comments.
\vspace{-25pt}}}

\numberofauthors{3}
\author{
  \alignauthor
  Alin Bostan\\
  \affaddr{Inria (France)}\\
  \email{alin.bostan@inria.fr}
  \alignauthor
  Gilles Christol\\
  \affaddr{IMJ (France)}\\
  \email{gilles.christol@imj-prg.fr}
  \alignauthor
  Philippe Dumas\\
  \affaddr{Inria (France)}\\
  \email{philippe.dumas@inria.fr}
}

\begin{document}
\CopyrightYear{2016}
\conferenceinfo{ISSAC'16,}{July 19--22, 2016, Waterloo, ON, Canada} \isbn{978-1-4503-4380-0/16/07}	
\doi{http://dx.doi.org/10.1145/2930889.2930904} 

\clubpenalty=10000 
\widowpenalty = 10000

\setlength{\belowdisplayskip}{.3\baselineskip} \setlength{\belowdisplayshortskip}{0pt}
\setlength{\abovedisplayskip}{.36\baselineskip} \setlength{\abovedisplayshortskip}{0pt}

\maketitle

\begin{abstract} We address the question of computing one selected term of an
algebraic power series. In characteristic zero, the best algorithm currently
known for computing the~$N$th coefficient of an algebraic series uses
differential equations and has arithmetic complexity quasi-linear
in~$\sqrt{N}$. We show that over a prime field of positive characteristic~$p$,
the complexity can be lowered to~$O(\log N)$. The mathematical basis for this
dramatic improvement is a classical theorem stating that a formal power series
with coefficients in a finite field is algebraic if and only if the sequence
of its coefficients can be generated by an automaton. We revisit and enhance
two constructive proofs of this result for finite prime fields. The first
proof uses Mahler equations, whose sizes appear to be prohibitively large. The
second proof relies on diagonals of rational functions; we turn it into an
efficient algorithm, of complexity linear in~$\log N$ and quasi-linear in~$p$.
\end{abstract}

\vspace{-1mm}
\ccsdesc[500]{Computing methodologies~Algebraic algorithms}
\printccsdesc

\vspace{-1.5mm}
\keywords{algebraic series; finite fields; Mahler equations; diagonals; $p$-rational series; section operators; algebraic complexity}

\section{Introduction}
 \begin{flushright}
   {{\em One of the most difficult questions in modular computations is the complexity of computations $\bmod ~p$ for a large prime~$p$ of coefficients in the expansion of an algebraic function.}
     \\ D.V. Chudnovsky \& G.V. Chudnovsky, 1990~\cite{ChudnovskyChudnovsky90}.}
\end{flushright}

\noindent{\bf Context.}
Algebraic functions are ubiquitous in all branches of pure and applied mathematics, notably in algebraic geometry, combinatorics and number theory.
They also arise at the confluence of several fields in computer science: functional equations, automatic sequences, complexity theory.
From a computer algebra perspective, a fundamental question is the 
efficient computation of power series expansions of algebraic functions.
We focus on the particular question of computing \emph{one selected term} of an
\emph{algebraic power series}~$f$ whose coefficients belong to a \emph{field of positive characteristic~$p$.} Beyond its relevance to complexity theory, this problem is important in applications to integer factorization and point-counting~\cite{BoGaSc07}.

\medskip\noindent{\bf Setting.}
More precisely, we assume in this article that the ground field is the prime field~$\bF_p$ 
and that the
power series~$f$ is 
(implicitly) given as the unique solution in ${\bF_p}[[x]]$~of 
\begin{equation*}\label{eq:hyp1}
  E(x,f(x)) = 0,\qquad f(0) = 0,
\end{equation*}
where~$E$ is a polynomial in ${\bF_p}[x,y]$
that satisfies 
\begin{equation*} \label{eq:hyp2}
  E(0,0) = 0\quad \text{and} \quad E_y(0,0) \neq 0.
\end{equation*}
(Here, and hereafter, $E_y$ stands for the partial derivative~$\frac{\partial E}{\partial y}$.)

Given as input the polynomial $E$ and an integer~$N>0$, our algorithmic
problem is to efficiently compute the $N$th coefficient~$f_N$ of $f=\sum_i f_i
x^i$. The efficiency is measured in terms of number of arithmetic operations
$(\pm, \times, \div)$ in the field~$\bF_p$, the main parameters being the
index~$N$, the prime~$p$ and the bidegree $(h, d)$ of $E$ with respect to
$(x,y)$.

In the particular case $d=1$, the algebraic function $f$ is actually a rational function, and thus the $N$th coefficient of its series expansion can be computed in $O(\log N)$ operations in ${\bF_p}$, using standard binary powering techniques~\cite{MillerSpencer66,Fiduccia85}.

Therefore, it will be assumed in all that follows that $d>1$.

\smallskip\noindent{\bf Previous work and contribution.}
The most straightforward method for computing the coefficient $f_N$ of the algebraic power series~$f$ proceeds by undetermined coefficients. Its arithmetic complexity is $O(N^d)$.
Kung and Traub~\cite{KungTraub} showed that the formal Newton iteration can accelerate this to~$\softO(dN)$. 
(The soft-O notation $\softO( \cdot)$ indicates that polylogarithmic factors are omitted.)
Both methods work in arbitrary characteristic and compute~$f_N$ together with all $f_i, i < N$.

In characteristic zero, it is possible to compute the coefficient $f_N$ faster, without computing all the previous ones. This result is due to the Chudnovsky brothers~\cite{ChudnovskyChudnovsky88} and is based on the classical fact that the coefficient sequence $(f_n)_{n\geq 0}$ satisfies a linear recurrence with polynomial coefficients~\cite{Comtet64,ChudnovskyChudnovsky86,BoChLeSaSc07}. 
Combined with baby steps/giant steps techniques, this 
leads to an algorithm of complexity quasi-linear in~$\sqrt{N}$.
Except for the very particular case of rational functions ($d=1$), no faster method is currently known.

Under restrictive assumptions, the baby step/giant step algorithm can be adapted to the case of positive characteristic~\cite{BoGaSc07}. The main obstacle is the fact that the linear recurrence satisfied by the sequence $(f_n)_{n\geq 0}$ has a leading coefficient that may vanish at various indices. In the spirit of~\cite[\S8]{BoGaSc07}, $p$-adic lifting techniques could in principle be used, but we are not aware of any sharp analysis of the sufficient $p$-adic precision in the general case. Anyways, the best that can be expected from this method is a cost quasi-linear in~$\sqrt{N}$. 

We attack the problem from a different angle. Our starting point is a theorem
due to the second author in the late 1970s~\cite[Th.~12.2.5]{AlSh03}. It
states that a formal power series with coefficients in a finite field is algebraic if and only if the sequence $(f_n)_{n \geq 0}$ of its coefficients can be
generated by a $p$-automaton, i.e., $f_N$ is the output of a finite-state machine taking as input the digits of $N$ in base~$p$.
This implicitly contains the roots of a $\log N$-method for computing $f_N$,
but the size of the $p$-automaton is at least $p^{(d+h)^2}$, see e.g.~\cite{RoYa15}.

In its original version~\cite{Christol79} the theorem was stated for~$\bF_2$, 
but the proof extends \emph{mutatis mutandis} to any finite field. 
A different proof was given by Christol, Kamae, Mendès-France and Rauzy~\cite[\S7]{ChKaMeRa80}. 
Although constructive in essence, these proofs do not focus on computational
complexity aspects. 

Inspired by~\cite{AlSh92} and~\cite{Dumas93}, we show that each of them leads to an algorithm of arithmetic complexity~$O(\log N)$ for the computation of~$f_N$, after a precomputation that may be costly for $p$ large. On the one hand, the proof in~\cite{ChKaMeRa80} relies on the fact that the sequence $(f_n)_{n \geq 0}$ satisfies a divide-and-conquer recurrence. 
However, we show (Sec.~\ref{sec:Mahler}) that the size of the recurrence is polynomial in $p^d$, making the algorithm uninteresting even for very moderate values of~$p$.
On the other hand, the key of the proof in~\cite{Christol79} is to represent~$f$ as the diagonal of a bivariate rational function.
We turn it (Sections~\ref{sec:diag}--\ref{sec:newalgo}) into an efficient algorithm that has complexity $O(\log N)$ after a precomputation whose cost is $\softO (p)$ only.
To our knowledge, the only previous explicit occurrence of a $\log N$-type complexity for this problem appears in~\cite[p.~121]{ChudnovskyChudnovsky90}, which announces the bound $O(p \cdot \log N)$ but without proof.

 \smallskip\noindent{\bf Structure of the paper.} In Sec.~\ref{sec:Mahler}, we propose an algorithm that computes the $N$th coefficient $f_N$ using Mahler equations and divide-and-conquer recurrences. Sec.~\ref{sec:diag} is devoted to the study of a different algorithm, based on the concept of diagonals of rational functions.  We conclude in Sec.~\ref{sec:newalgo} with the design of our main algorithm. 

\smallskip\noindent{\bf Cost measures.}
We use standard complexity notation. The real number $\omega>2$ denotes a feasible
exponent for matrix multiplication i.e., there exists an algorithm for
multiplying $n \times n$ matrices with entries in~${\bF_p}$ in
$O(n^\omega)$ operations in~${\bF_p}$. 
The best bound currently known is $\omega <2.3729$
from~\cite{LeGall14}. 
We use the fact that
many arithmetic operations in ${\bF_p}[x]_{d}$, the set
of polynomials of degree at most~$d$ in ${\bF_p}[x]$, can be performed in~$\softO(d)$ operations: addition, multiplication, division, 
\emph{etc}. The key to these results is a divide-and-conquer approach combined with fast polynomial
multiplication~\cite{Schoenhage77,CaKa91,HaHoLe14}. 
A general reference on fast algebraic algorithms is~\cite{GathenGerhard2013}.

\section{Using Mahler Equations}\label{sec:Mahler}
We revisit, from an algorithmic point of view, the proof in~\cite[\S7]{ChKaMeRa80} of the fact that the
coefficients of an algebraic power series $f\in{\bF_p}[[x]]$ can be recognized by
a $p$-automaton.
Starting from an algebraic equation for~$f=\sum_n f_n x^n$, we first compute a Mahler equation satisfied by~$f$, then we derive an appropriate divide-and-conquer recurrence for its coefficient sequence $(f_n)_n$, and use it to compute~$f_N$ efficiently. 
We will show that, although the complexity of this method is very good (i.e., logarithmic) with respect to the index~$N$, the computation cost of the recurrence, and actually its mere size, are extremely high (i.e., exponential) with respect to the algebraicity degree $d$.

\subsection{From algebraic equations to Mahler\\ equations and DAC recurrences}\label{subsec:AlgEqToMahlerEq}
\emph{Mahler equations} are functional equations which use the Mahler operator~$M_p$, or~$M$ for short, acting on formal power series by substitution: $M g(x) = g(x^p)$. Their power series solutions are called \emph{Mahler series}.
In characteristic~$0$, there is no a priori link between algebraic equations and Mahler equations; moreover, a series which is at the same time algebraic and Mahler is a rational function~\cite[\S1.3]{Nishioka1996}. 

By contrast, over $\bF_p$, a Mahler equation is nothing but an algebraic equation, because of the Frobenius endomorphism that permits writing $g(x^p) = g(x)^p$ for any $g\in{\bF_p}[[x]]$. 
Hence a Mahler series in ${\bF_p}[[x]]$ is an algebraic series. Conversely, an algebraic series is Mahler, since if~$f$ is algebraic of degree~$d$ then the $(d+1)$ power series $f$, $Mf=f^p$, \ldots, $M^d f = f^{p^d}$ 
are linearly dependent over~${\bF_p}(x)$.

Concretely, the successive powers~$f^{p^k}$ ($0 \leq k \leq d$) are expressed as linear combinations of~$1$, $f$, $\ldots$, $f^{d-1}$ over~${\bF_p}(x)$, and any dependence relation between them 
delivers a nontrivial Mahler equation with polynomial coefficients 
\begin{equation}\label{IssacSubmission:eq:MahlerEquation}
  c_0(x) f(x) + c_1(x) f(x^p) + \dotsb + c_{\MhlEqOrder}(x) f(x^{p^{\MhlEqOrder}}) = 0,
\end{equation}
whose order~$\MhlEqOrder$ is at most~$d$. 

\begin{example}[A toy example]
Consider $E = x + y - y^3$ in $\bF_5[x,y]$. 
Let $f = -x-x^3+2x^5-2x^7+2x^{11}+\cdots$ be the 
unique solution in $x\,\bF_5[[x]]$ of $E(x,f(x))=0$.
The expressions of~$f$, $f^5$, $f^{25}$, $f^{125}$ as linear combinations of~$1$, $f$, $f^2$ with coefficients in~$\bF_5(x)$ give the columns of  the matrix
\begin{equation*}\label{IssacSubmission:eq:MhlMatrix}
\left[ \begin {array}{cccc} 0 &x&-2\,{x}^{7}+{x}^{5}+x&{x}^{41}-{x}^{
37} + \dotsb 
\\ \noalign{\medskip}1&1&{x}^{8}-{x}^{6}+1&{x}^{40}-2\,{x}^{38}+ \dotsb 
\\ \noalign{\medskip}0&x&-2\,{x}^{7}+{x}^{5}+x&{x}^{41}-{x}^{37} + \dotsb \end {array}
 \right] 
 .
\end{equation*}
The first three columns are linearly dependent and lead to the Mahler equation
\begin{equation}\label{IssacSubmission:eq:MereExampleMahlerEquation}
  x^4(1 - x^2 - x^4) f(x) - (1 + x^4 - 2 x^6) f(x^5) + f(x^{25}) = 0.
\end{equation}
\end{example}

A Mahler equation instantly translates into a recurrence of  divide-and-conquer type, in short a DAC recurrence. Such a recurrence links the value of the index-$n$ coefficient~$f_n$ with some values for the indices~$n/p$, $n/p^2$\ldots, and some shifted indices. The translation is absolutely simple: a term $x^s f(x)$ in equation~\eqref{IssacSubmission:eq:MahlerEquation} translates into $f_{n-s}$; 
more generally, a term $x^s f(x^q)$ becomes $f_{(n-s)/q}$, with the convention that a coefficient~$f_\nu$ is zero if~$\nu$ is not a nonnegative integer.

\begin{example}[A binomial case]\label{IssacSubmission:example:ABinomialCase}
Let $p>2$ be a prime number and let 
$f = x + x^2 + \cdots $ be the algebraic series in~$x\,\bF_p[[x]]$ solution of $E(x,y) = x + (1 + y)^{p-1} - 1$. The rewriting
\[
  x + (1 + f)^{p-1} - 1 = x + \frac{1 + f^p}{1 + f} - 1 = \frac{(x-1)(1+f) + (1 + f^p)}{1 + f}
\]
yields $f^p = - x + (1 - x) f$ and $f^{p^2} = -x^p + (1 - x^p) f^p $; hence $f^{p^2} = (1-x)^{p+1} - 1 + (1-x)^{p+1}f$. The situation is highly non-generic, since all powers~$f^{p^k}$ are expressible as linear combinations of~$1$ and~$f$ only. 
This delivers the second-order Mahler equation
\begin{equation*}\label{IssacSubmission:eq:ABinomialCaseMhlEq}
  (x^{p-1} - x^{p})\, f(x) - (1 +  x^{p-1} - x^{p})\, f(x^p) + f (x^{p^2})  = 0,
\end{equation*}
which translates into a recurrence on the coefficients sequence
\begin{equation}\label{IssacSubmission:eq:ABinomialCaseDACRec}
  f_{n-p+1} - f_{n-p} - f_{\frac{n}{p}} - f_{\frac{n-p+1}{p}} + f_{\frac{n-p}{p}} + f_{\frac{n}{p^2}} = 0.
\end{equation}
It is possible to make the recurrence more explicit by considering the~$p^2$ cases according to the value of the residue of~$n$ modulo~$p^2$. 
\end{example}

\subsection{Nth coefficient via a DAC recurrence}\label{IssacSubmission:sec:NthCoefficientViaDACRecurrence}
The DAC recurrence relation attached to the Mahler
equation~\eqref{IssacSubmission:eq:MahlerEquation}, together with enough
initial conditions, can be used to compute the $N$th coefficient~$f_N$ of the
algebraic series~$f$. The complexity of the resulting algorithm is linear with
respect to~$N$. The reason is that the coefficient~$c_0(x)$
in~\eqref{IssacSubmission:eq:MahlerEquation} generally has more than one
monomial, and as a consequence, all the coefficients $f_i$, $i< N$ are needed 
to compute~$f_N$.

\begin{example}[A binomial case, cont.]
We specialize Ex.~\ref{IssacSubmission:example:ABinomialCase} by taking $p = 7$ and compute~$f_{N}$ with $N = 100$. Applying~\eqref{IssacSubmission:eq:ABinomialCaseDACRec} to $n = 106$, we need to know the value for $n = 99$, next for $n = 98$, $n = 15$, $n = 14$, and also for $n = 97$, $\ldots$, $n = 91$, $n = 14$, $n = 13$, $n = 2$. 
In the end, it appears that all values of $f_n$ for the indices $n$ smaller than~$N = 100$ are needed to compute~$f_{N}$.
\end{example}

It is possible to decrease dramatically the cost of the computation of $f_N$, from linear in~$N$ to logarithmic in~$N$. The idea is to derive from~\eqref{IssacSubmission:eq:MahlerEquation} another Mahler equation with the additional feature that its trailing coefficient $c_0(x)$ is~1. 
To do so, it suffices to perform a change of unknown series. Putting $f(x) = c_0(x) g(x)$, we obtain the Mahler equation 
\begin{multline}\label{IssacSubmission:eq:MonicMahlerEquation}
  g(x) + c_1(x)c_0(x)^{p-2} g(x^p) + \dotsb \\
  \mbox{} + c_{\MhlEqOrder}(x)c_0(x)^{p^{\MhlEqOrder} -2}g(x^{p^{\MhlEqOrder}}) = 0.
\end{multline}
Obviously this approach assumes that~$c_0$ is not zero. But, as proved in~\cite[Lemma 12.2.3]{AlSh03}, this is the case 
if~\eqref{IssacSubmission:eq:MahlerEquation} is assumed to be
a \emph{minimal-order} Mahler equation satisfied by~$f$.

The series~$g(x)$ is no longer a power series, but a Laurent series in~${\bF_p}((x))$. We cut it into two parts, its negative part~$g_-(x)$ in ${\bF_p}[x^{-1}]$ and its nonnegative part~$h(x)$ in~${\bF_p}[[x]]$:
\begin{equation*}
  g_-(x) = \sum_{n < 0} g_n x^n,\qquad h(x) = \sum_{n \geq 0} g_n x^n.
\end{equation*}

Let us rewrite Eq.~\eqref{IssacSubmission:eq:MonicMahlerEquation} in the compact form $L(x,M) g(x) = 0$, where~$L(x,M)$ is a skew polynomial in the variable~$x$ and in the Mahler operator~$M$. Plugging $g(x) = g_-(x) + h(x)$ in it yields three terms. The first is the power series $L(x,M) h(x)$. The second is the nonnegative part~$-b(x)$ of $L(x,M) g_-(x)$, in~${\bF_p}[x]$. The third is the negative part of $L(x,M) g_-(x)$, and it is zero because it is the only one which belongs to~$x^{-1} {\bF_p}[x^{-1}]$.  Eq.~\eqref{IssacSubmission:eq:MonicMahlerEquation} is rewritten as a new (inhomogeneous) Mahler equation $L(x,M) h(x) = b(x)$, namely
\begin{multline}\label{IssacSubmission:eq:NonHomogeneousMonicMahlerEquation}
  h(x) + c_1(x)c_0(x)^{p-2} h(x^p) + \dotsb \\
  \mbox{} + c_{\MhlEqOrder}(x)c_0(x)^{p^{\MhlEqOrder} -2}h(x^{p^{\MhlEqOrder}}) = b(x).
\end{multline}

\begin{example}[A toy example, cont.]
\  The change of series $f(x) = x^4 (1 - x^2 - x^4) h(x)$ in~\eqref{IssacSubmission:eq:MereExampleMahlerEquation} gives the new equation
\begin{multline*}
  h(x) - x^{12} (1-x)(1+x)(1+x^2+2x^4)(1 - x^2 - x^4)^3 h(x^5) \\
  \mbox{}+ x^{92} (1 - x^2 - x^4)^{23} h(x^{25}) = \\
  -x-{x}^{5}+{x}^{7}+2\,{x}^{9} + \dotsb 
  +{x}^{149}-{x}^{157}-2\,{x}^{159}
\end{multline*}
and a recurrence
\[
  h_n = h_{\frac{n-12}{5}} + 2 h_{\frac{n-14}{5}} + \dotsb - h_{\frac{n-92}{25}} + \dotsb,
\]
while~$f_n$ is given by
$
  f_n = h_{n-4} - h_{n-6} - h_{n-8}\quad\text{for $n \geq 8$}.
$

The right-hand side $b(x)$ of this new inhomogeneous equation is obtained as the nonnegative part of $-(c_1 c_0^3 g_-^5 + c_2 c_0^{23} g_-^{25})$, where $c_0=x^4(1-x^2-x^4)$, $c_1=2x^6-x^4-1$, $c_2=1$ are the coefficients of Eq.~\eqref{IssacSubmission:eq:MereExampleMahlerEquation}, and where $g_- = -2x^{-1}-x^{-3}$, the negative part of $f/c_0$, can be determined starting from $(f \bmod x^4)$.

To compute the coefficient~$f_{N}$ for $N=1251$, this approach only requires the computation of thirteen terms of the sequence~$h_n$, namely for $n \in \{0, 3, 5, 7, 43, 45, 47, 243, 245, 247, 1243, 1245, 1247\}$. This number of terms behaves like $3 \log_p N$, and compares well to the use of a recurrence like~\eqref{IssacSubmission:eq:ABinomialCaseDACRec}, which requires $N$ terms.
\end{example}
At this point, it is intuitively plausible that the $N$th coefficient $f_N$ can be computed in arithmetic complexity~$\bigO(\log N)$: to obtain~$f_N$, we need a few values~$h_N$ and these ones are essentially obtained from a bounded number of~$h_{N/p}$, $h_{N/p^2}$\ldots. 

\subsection{Nth coefficient via the section operators}
This plausibility can be strengthened and made clearer with the introduction of the \emph{section operators}
(sometimes called \emph{Cartier operators}, or \emph{decimation operators}) 
and the concept of \emph{$p$-rational series}. 

\begin{definition}
The section operators $S_0, \ldots, S_{p-1}$, with respect to the radix~$p$, are defined by
\begin{equation}\label{IssacSubmission:eq:SectionOperatorsDefinition}
  S_r \sum_{n \geq 0} f_n x^n = \sum_{k \geq 0} f_{pk + r} x^k,\qquad \text{for} \quad 0 \leq r < p.
\end{equation}
\end{definition}
In other words, there is a section operator for each digit of the radix-$p$ numeration system and the $r$th section operator extracts from a given series its part associated with the indices congruent to the digit~$r$ modulo~$p$.  These operators are linear and well-behaved with respect to the product:
\begin{equation*}
  S_r [f(x)g(x)] = \sum_{s + t \equiv r \bmod p} x^{\lfloor \frac{s+t}{p} \rfloor} S_s f(x) S_t g(x).
\end{equation*}
In particular, for $g(x) = h(x^p)$ the previous formula becomes
\begin{equation}\label{IssacSubmission:eq:SectionOperatorsAndProduct}
  S_r [f(x)h(x^p)] = S_r [f(x)] \times h(x), 
\end{equation}
because of the obvious relationships with the Mahler operator $S_0M = \Id_{{\bF_p}[[x]]}$, $S_r M = 0$ for $r > 0$. The action of the section operators can be extended to the field~${\bF_p}((x))$ of formal Laurent series and the same properties apply to the extended operators, which we denote in the same way.

The section operators permit to express the coefficient~$h_N$ of a formal power series~$h(x)$ using the radix-$p$ digits of~$N$.

\begin{lem}
Let $h$ be in ${\bF_p}[[x]]$ and let 
  $N = (N_{\ell}\dotsb N_1 N_0)_p$
be the radix-$p$ expansion of $N$.
Then
\begin{equation}\label{IssacSubmission:eq:NthCoefficientAsAnOperatorProduct}
  h_N = (S_{N_{\ell}}\dotsb S_{N_1} S_{N_0} h)(0).
\end{equation}
\end{lem}

\begin{proof}
We first apply the section operator~$S_{N_0}$ to the series~$h(x)$ associated with the least significant digit of $N$, that is we start from~$N_0$ and we pick every $p$th coefficient of~$h(x)$. This gives
 $ S_{N_0} h(x) = h_{N_0} + h_{N_0 + p} x + h_{N_0 + 2 p} x^2 + \dotsb. $
Iterating this process with each digit produces the series
\begin{equation*}
  S_{N_{\ell}}\dotsb S_{N_1} S_{N_0} h(x) = h_N + h_{N+p^{\ell +1}} x + \dotsb,
\end{equation*}
whose constant coefficient is $h_N$.
\end{proof}

\bigskip
\subsection{Linear representation}\label{ssec:linrep}
Let us return to the algebraic series~$f(x)$ and its relative~$h(x)$. The Mahler equation~\eqref{IssacSubmission:eq:NonHomogeneousMonicMahlerEquation} can be rewritten
\begin{equation}\label{eq:Mahler-monic}
  h(x) = b(x) +  a_1(x) h(x^p) + \dotsb + a_{\MhlEqOrder}(x) h(x^{p^{\MhlEqOrder}}).
\end{equation}
The coefficients~$b(x)$ and~$a_k(x)$, $1 \leq k \leq  {\MhlEqOrder}$, are polynomials of degrees at most~$D$, say. As a matter of fact, the vector space~$\mathcal{W}$ of linear combinations
\[
  s(x) = a(x) + b_0(x) h(x) + b_1(x) h(x^p) + \dotsb + b_{\MhlEqOrder}(x) h(x^{p^{\MhlEqOrder}})
\]
with~$a(x)$ and all~$b_k(x)$ in~${\bF_p}[x]_{D}$ is stable under the action of the section operators. Indeed, from 
\begin{multline*}
  s(x) =  (a(x) + b_0(x) b(x)) +  (b_0(x) a_1(x) + b_1(x)) h(x^p) + \mbox{} \\  \dotsb + (b_0(x) a_{\MhlEqOrder}(x) + b_{\MhlEqOrder}(x)) h(x^{p^{\MhlEqOrder}})
\end{multline*}
we deduce 
\begin{multline}\label{IssacSubmission:eq:ExplicitSectionOperators}
   S_r s(x) = 
   S_r \left(a(x) + b_0(x) b(x)\right) \\ \mbox{} + 
   S_r \left(b_0(x) a_1(x) + b_1(x)\right) h(x) + \dotsb \\ \mbox{} + S_r \left(b_0(x) a_{\MhlEqOrder}(x) + b_{\MhlEqOrder}(x)\right) h(x^{p^{{\MhlEqOrder} -1 }}),
\end{multline}
and the polynomials $S_r (a + b_0 b)$ and $S_r(b_0 a_k + b_k)$, $1 \leq k \leq  {\MhlEqOrder}$, have degrees not greater than~$2 D / p \leq D$. 

The important consequence of this observation is that we have produced a
finite dimensional ${\bF_p}$-vector space~$\mathcal{W}$ that contains~$h(x)$
and that is stable under the sections $(S_r)_{0\leq r < p}$. This will
enable us to effectively compute the coefficient $h_N$ using
formula~\eqref{IssacSubmission:eq:NthCoefficientAsAnOperatorProduct}, by
performing matrix computations.

The vector space~$\mathcal{W}$ is spanned over $\bF_p$ by the family
$\mathcal{F} = \big\{ x^i, 0\leq i\leq D\big\} \cup \big\{ x^j h(x^{p^k}),
0\leq j \leq D, 0\leq k \leq K \big\}$. Let~$A_r$ be a matrix of~$S_r$ with respect to~$\mathcal{F}$, let~$C$ be the column vector containing only zero entries except an entry equal to 1 at index $(j,k)=(0,0)$, and~$L$ be the row matrix of the evaluations at~$0$ of the elements of~$\mathcal{F}$.

With these notation, Eq.~\eqref{IssacSubmission:eq:NthCoefficientAsAnOperatorProduct} rewrites in matrix terms:
\begin{equation}\label{IssacSubmission:eq:NthCoefficientAsAMatrixProduct}
  h_N = L A_{N_{\ell}}\dotsb A_{N_1} A_{N_0} C.
\end{equation}

\begin{definition}
The family of matrices $\big( L$, $(A_r)_{0 \leq r < p}, C\big)$ is a 
\emph{linear representation} of the series~$h(x)$. 
A power series 
that admits a linear representation is said to be \emph{$p$-rational}. 
\end{definition}

As it was pointed out in~\cite[p.~178]{AlSh92}, an immediate consequence of
Eq.~\eqref{IssacSubmission:eq:NthCoefficientAsAMatrixProduct} is that the
coefficient $h_N$ can be computed using only $O(\log N)$ matrix-vector
products in size $(D+1)(K+1)$, therefore in arithmetic complexity $O((DK)^2
\log N)$.

Actually, the matrices $A_r$ are structured, and their product by a vector can be done faster than quadratically.
Indeed, Eq.~\eqref{IssacSubmission:eq:ExplicitSectionOperators} shows that one can perform a matrix-vector product by the matrix~$A_r$ using $K$ polynomial products in ${\bF_p}[x]_{D}$, thus in quasi-linear complexity $\softO (DK)$.
We deduce:

\begin{prop}\label{prop:Nth-term-Mahler-monic}
The $N$th coefficient $h_N$ of the solution~$h$ of~\eqref{eq:Mahler-monic}
can be computed in time $\softO((DK) \log N)$.
\end{prop}

\subsection{Nth coefficient using Mahler equations
} \label{ssec:algo-via-Mahler}
We are ready to prove the main result of this section.
\begin{thm}\label{IssacSubmission:thm:CKMRComplexity}
Let $E$ be a polynomial in {$\bF_p[x,y]_{h,d}$} such that $E(0,0) =
0$ and $E_y(0,0) \neq 0$, and let $f\in\bF_p[[x]]$ be its unique root with $f(0)=0$.
Algorithm~\ref{IssacSubmission:algo:CKMR} computes 
the $N$th coefficient~$f_N$ of~$f$ using $\softO(d^3 h^2 p^{3d} \log N)$ operations in~$\bF_p$. 
\end{thm}

\begin{algo}	
	Algorithm \textbf{Nth coefficient via Mahler equations}
	
	\begin{algoenv}{A polynomial $E(x,y)\in \bF_p[x,y]$ with 
$E(0,0) = 0$ and $E_y(0,0) \neq 0$, and an integer~$N\geq 0$.}
{The $N$th coefficient~$f_N$ of the unique formal power series~$f\in\bF_p[[x]]$ solution of $E(x,f(x)) = 0$, $f(0) = 0$.}
    \State1. Use Algorithm~\ref{algo:AlgeqToMahler}
 to compute a minimal-order Mahler equation   
(with $c_0 = c_{0,v_0} x^{v_0} + \dotsb + c_{0,d_0}x^{d_0}, \, c_{0,v_0}\neq 0$):

$L(x,M) f(x) = c_0(x) f(x) + \dotsb + c_K(x) f(x^{p^K}) = 0$.

    \State2. Deduce an equation $L'(x,M) g(x) = 0$~\eqref{IssacSubmission:eq:MonicMahlerEquation}, with $c_0' = 1$, satisfied by $g(x) = f(x)/c_0(x)$.
    \State3. Compute $f(x) \bmod x^{v_0}$ by Newton iteration, and deduce the negative part~$g_-(x)$ of $f(x)/c_0(x)$.
    \State4. Compute the nonnegative part~$b(x)$ of $- L'(x,M) g_-(x)$

    \Comment The nonnegative part $h(x)$ of $f(x)/c_0(x)$ satisfies 

	\quad \; the equation $L'(x,M) h(x) =  b(x)$~\eqref{IssacSubmission:eq:NonHomogeneousMonicMahlerEquation}.

    \State5. Get a linear representation~$\left(L, (A_r)_{0 \leq r < p}, C\right)$ for~$h(x)$.

  \State 6. {\textbf{for}} {$N'$ in $\{ N-d_0, \ldots, N-v_0 \}$}
   \State \rule{2.3em}{0ex} write $N' = (N'_{\ell} \ldots N'_0)_p$ and 
  \State  \rule{2.3em}{0ex} compute~$h_{N'} = L A_{N'_{\ell}}\dotsb A_{N'_0}C$;

  \State7. \Return $f_N = c_{0,v_0} h_{N-v_0} + \dotsb + c_{0,d_0} h_{N- d_0 }$.
	\end{algoenv}
	\caption{\label{IssacSubmission:algo:CKMR}$N$th coefficient via Mahler equations.}
\end{algo}

\begin{algo}[ht]
 Algorithm \textbf{AlgeqToMahler} 

 	\begin{algoenv}{$E(x,y)\in {\bF_p}[x,y]$ 
		with $E(0,0) = 0$ and $E_y(0,0) \neq 0$.
	}
	{A minimal-order monic Mahler operator for the unique $f \in {\bF_p}[[x]]$ with $E(x,f(x))=0$ and $f(0)=0$.}

    \State $R_0=y$;
    \For {$s =  1,2,\dots$}
      \State $R_s = R_{s-1}^{p} \bmod E$ \quad \Comment{$R_s = y^{p^s} \bmod E$}
      \If {$\operatorname{rank}_{{\bF_p}(x)}(R_0,R_1,\dots,R_s)<s+1$}
        \State Solve $\sum_{k=0}^{s-1} \frac{c_k}{c_s} R_k=-R_s$ for $c_0,\dots,c_{s}$ in ${\bF_p}[x]$ 
        \State \Return $\sum_{k=0}^{s} {c_k} M^{k}$
      \EndIf  
    \EndFor
  \end{algoenv}
  \caption{From algebraic equations to Mahler equations.}
\label{algo:AlgeqToMahler}
\end{algo}

To do this, we first need a preliminary result that will also be useful in Sec.~\ref{sec:newalgo}.
It provides size and complexity bounds for the remainder of the Euclidean division of a monomial in~$y$ by a polynomial in~$y$ with coefficients in ${\bF_p}[x]$. 
\begin{lem}\label{lem:boundsMahler}
Let $E = \sum_i e_i(x) y^i$ be a polynomial in ${\bF_p}[x][y]$ of degree~$d$ in $y$ and at most $h$ in $x$. Then, for $D\geq d$,
\[y^D \bmod E = \frac{1}{e_d^{D-d+1}} \left( r_0(x) + \cdots + r_{d-1}(x)y^{d-1}\right),\]
where the $r_i$'s are in ${\bF_p}[x]$ of degree at most $h(D-d+1)$.

One can compute the $r_i$'s using $\softO(h d D)$ operations in ${\bF_p}$.
\end{lem}

\begin{proof}
The degree bounds are proved by induction on~$D$.
The computation of $y^D \bmod E$ can be performed by binary powering. The last step dominates the cost of the whole process. It amounts to first computing the product of two polynomials in ${\bF_p}[x,y]$ of degree at most $d-1$ in $y$ and of degree $O(hD)$ in $x$, then reducing the result modulo $E$. Both steps have complexity $\softO(h d D)$: the multiplication is done via Kronecker's substitution~\cite[Cor.~8.28]{GathenGerhard2013}, the reduction via an adaptation of the Newton iteration used for the fast division of univariate polynomials~\cite[Th.~9.6]{GathenGerhard2013}.
\end{proof}

The first step of Algorithm~\ref{IssacSubmission:algo:CKMR} is the computation of a Mahler equation for the algebraic series $f$. We begin by analyzing the sizes of this equation, and the cost of its computation. 

\begin{thm}\label{thm:AlgeqToMahler}
Let $E$ be a polynomial in $\bF_p[x,y]_{h,d}$ such that $E(0,0) =
0$ and $E_y(0,0) \neq 0$, and let $f\in\bF_p[[x]]$ be its unique root with $f(0)=0$.
Algorithm~\ref{algo:AlgeqToMahler} computes a Mahler equation of minimal-order satisfied by $f$ 
using $\softO(d^\omega h p^{d})$ operations in~$\bF_p$.
The output equation has order at most $d$ and polynomial coefficients in $\bF_p[x]$ of degree at most $dhp^d$.
\end{thm}

\begin{proof}
We first prove the correctness part. Since the $R_i$'s all live in the vector space generated by $1, f,\ldots, f^{d-1}$ over ${\bF_p}(x)$, the algorithm terminates and returns a Mahler operator of order $K$ at most $d$.

At step $s$ we have 
$R_s = y^{p^s} \bmod E$, thus $R_s(x,f) = f^{p^s}$. Since 
$\sum_{k=0}^K c_k R_k = 0$, this yields $L (f) = \sum_{k=0}^K c_k f^{p^k} = 0$, so $L$ is a Mahler operator for $f$. It has minimal order, because otherwise $R_0, \ldots, R_{K-1}$ would be linearly dependent over~${\bF_p}(x)$.
Therefore, the algorithm is correct. 

By Lemma~\ref{lem:boundsMahler}, $e_d^{p^s-d+1} R_s$ is in ${\bF_p}[x,y]$ of degree in $y$ at most $d-1$ and degree in $x$ at most $h (p^s -d + 1)$. Moreover, $R_0, \ldots, R_K$ can be computed in time $\softO(hd p^K)$.

By Cramer's rule, the degrees of the $c_k$'s are bounded by $dhp^K \leq dh p^d$. Testing the rank condition, and solving for the $c_k$'s amounts to polynomial linear algebra in size at most $d \times (d+1)$
and degree at most $hp^d$. 
This can be done in complexity
$\softO(d^\omega h p^{d})$, using Storjohann's algorithm~\cite{Sto03}.
\end{proof}	

We are now ready to prove Theorem~\ref{IssacSubmission:thm:CKMRComplexity}.

\begin{proof}[of Theorem~\ref{IssacSubmission:thm:CKMRComplexity}]
The correctness of Algorithm~\ref{IssacSubmission:algo:CKMR} follows from the discussion in Sec.~\ref{ssec:linrep}. 
By Theorem~\ref{thm:AlgeqToMahler}, the output of Step 1 is a Mahler equation of order $K \leq d$ and coefficients $c_j$ of maximum degree~$H \leq dhp^K$. In particular, $d_0 \leq dhp^d$. The cost of Step~1 is $\softO(d^\omega h p^d)$. Step 2 consists in computing the coefficients $c_j' = c_j c_0^{p^j-2}$ for $1\leq j \leq K$, of maximum degree $D \leq H p^K$. It can be performed in time $\softO(KH p^{K}) = \softO(d^2 h  p^{2d})$. Step~3 has negligible complexity $\softO(dv_0) = \softO(d^2h p^d)$ using the Kung-Traub algorithm~\cite{KungTraub}. Step 4 requires $\softO(KH p^{K}) = \softO(d^2 h  p^{2d})$ operations. By Proposition~\ref{prop:Nth-term-Mahler-monic} the last steps of the algorithm can be done using $\softO((DK) \log N) = \softO((d^2 h p^{2d}) \log N)$ for each $N'$, thus for a total cost of $\softO((d^3 h^2 p^{3d}) \log N)$. This dominates the whole computation.\end{proof}

Generically, the size bounds in Theorem~\ref{IssacSubmission:thm:CKMRComplexity} are quite tight.
This is illustrated by the following example.

\begin{example}[A cubic equation]\label{IssacSubmission:example:ACubicEquation} 
Let us consider the algebraic equation $E(x,y) = x - (1+x)y + x^2 y^2 + (1+x) y^3 = 0$ in~$\bF_3[x,y]$. The assumption $E_y(0,0) = -1 \neq 0$ is fulfilled. We find a Mahler equation of order~$3$, whose coefficients~$c_0 = 
{x}^{10}-{x}^{11}+
\dotsb
-{x}^{40}+{x}^{45}
$, $c_1$, $c_2$, $c_3$ have respectively heights~$45$, $47$, $50$, $32$. We compute  within precision~$\bigO(x^{10})$ 
the solution $f(x) = x-{x}^{2}-{x}^{3}+{x}^{5}+
\dotsb 
-{x}^{44}+\dotsb 
$,
and the negative part $g_-(x) = {x}^{-9}+{x}^{-7}+{x}^{-6}-{x}^{-4}-{x}^{-2}$ of $g = f/c_0$. We find a new operator  $L'(x,M)$ by computing the product $L(x,M)c_0(x)$ in~$\bF_3[x,M]$ and next a new equation $L'(x,M) y(x) = b(x)$ for the nonnegative part~$h(x)$ of~$g(x)$ by computing the polynomial $b(x) = L'(x,M) g_-(x)$. It appears that the coefficients~$c_0' = 1$, $c_1'$, $c_2'$, $c_3'$ and the right-hand side~$b(x)$ have  respectively heights~$0$, $92$, $365$, $1157$, and~$1103$. We arrive at a linear representation with size $ (1 + 3 + 1)\times (1 + 1157) = 5790$ for $p = 3$. If we increase the prime number~$p$, we successively find sizes $155190$ for $p = 5$, and~$1342725$ for $p = 7$. We do not pursue further because it is clear that such sizes make the computation unfeasible. 
\end{example}

\section{Using Diagonals}\label{sec:diag}

To improve the arithmetic complexity with respect to the prime number~$p$, we  need a different way to represent the algebraic series $f$. It is given by Furstenberg's theorem, and relies on the concept of \emph{diagonal of rational functions}.

\subsection{From algebraic equation to diagonal}
Furstenberg's theorem~\cite{Furstenberg67} tells us that every algebraic series is the diagonal of a bivariate rational function
\begin{equation}\label{IssacSubmission:eq:RationalFunctionDiagonal}
  f(x) = D \frac{a(x,y)}{b(x,y)},\qquad b(0,0) \neq 0.
\end{equation}
This means that $f(x) = \sum_i c_{i,i} x^i$, where $a/b = \sum_{i,j} c_{i,j} x^i y^j$.

Moreover, under the working assumption $E_y(0,0) \neq 0$, the rational function~$a/b$ is explicit in terms of $E$:
\begin{equation*}
  a(x,y) = y E_y(xy,y), \qquad b(x,y) = E(xy,y)/y.
\end{equation*}
(Notice that $b(0,0) = E_y(0,0) \neq 0$.)
If the polynomial $E(x,y)$ has degree~$d$ and height~$h$, then the partial degrees of $a$ and~$b$ are bounded as follows
\begin{multline}\label{IssacSubmission:eq:MaxPartialDegreesab}
  d_x = \max(\deg_x a,\deg_x b) \leq h,\\
  d_y = \max(\deg_y a,\deg_y b) \leq d + h.
\end{multline}

\begin{example}[A quartic equation]\label{IssacSubmission:example:AQuarticEquation}
Consider
\[
  f(x) = x+{x}^{3}+{x}^{4}+3{x}^{5}+5{x}^{6}+2{x}^{7}+4{x}^{8}-{x}^{9}
-{x}^{10}-2{x}^{11}+ \dotsb,
\]
the unique solution in~$x\,\bF_{11}[[x]]$ of the algebraic equation
\[
  E(x,y) = -x  + (1+x) y - (1+x^2) y^2 - y^3 + (1+x) y^4 = 0.
\]
Then $f$ is the diagonal of the rational function $a/b$ with
\[
  a = y- (2-x) {y}^{2} -3 {y}^{3}+  (4-2{x}^{2}) {y}^{4}+ 4x{y}^{5},
\]
\[
  b = (1-x)- (1-x) y - {y}^{2}+ (1-{x}^{2}) {y}^{3}+ x {y}^{4}.
\]
Fig.~\ref{IssacSubmission:fig:Diagonal} shows pictorially that~$f(x)$ is the diagonal of~$a/b$.

\begin{figure}[t]
\begin{center}
  \begin{picture}(190,95)(0,20)
    \put(-10,105){\tiny$\displaystyle \frac{y- (2-x) {y}^{2} -3 {y}^{3}+  (4-2{x}^{2}) {y}^{4}+ 4x{y}^{5}}{(1-x)- (1-x) y - {y}^{2}+ (1-{x}^{2}) {y}^{3}+ x {y}^{4}} = \mbox{}$}
  \put(175,30){    \makebox[3.0 \unitlength]{\tiny$\mbox{}\color{purple}+ 2{x}^{7}{y}^{7}
  $}}  \put(175,40){    \makebox[3.0 \unitlength]{\tiny$\mbox{}+ \phantom{4}{x}^{7}{y}^{6}
  $}}  \put(175,50){    \makebox[3.0 \unitlength]{\tiny$\mbox{}+ \phantom{4}{x}^{7}{y}^{5}
  $}}  \put(150,40){    \makebox[3.0 \unitlength]{\tiny$\mbox{}\color{purple}+ 5{x}^{6}{y}^{6}
  $}}  \put(125,30){    \makebox[3.0 \unitlength]{\tiny$\mbox{}+ 4{x}^{5}{y}^{7}
  $}}  \put(175,60){    \makebox[3.0 \unitlength]{\tiny$\mbox{}+ 4{x}^{7}{y}^{4}
  $}}  \put(150,50){    \makebox[3.0 \unitlength]{\tiny$\mbox{}-4{x}^{6}{y}^{5}
  $}}  \put(100,30){    \makebox[3.0 \unitlength]{\tiny$\mbox{}+ 2{x}^{4}{y}^{7}
  $}}  \put(175,70){    \makebox[3.0 \unitlength]{\tiny$\mbox{}+ 5{x}^{7}{y}^{3}
  $}}  \put(150,60){    \makebox[3.0 \unitlength]{\tiny$\mbox{}+ 3{x}^{6}{y}^{4}
  $}}  \put(125,50){    \makebox[3.0 \unitlength]{\tiny$\mbox{}\color{purple}+ 3{x}^{5}{y}^{5}
  $}}  \put(100,40){    \makebox[3.0 \unitlength]{\tiny$\mbox{}-3{x}^{4}{y}^{6}
  $}}  \put(75,30){    \makebox[3.0 \unitlength]{\tiny$\mbox{}+ 4{x}^{3}{y}^{7}
  $}}  \put(150,70){    \makebox[3.0 \unitlength]{\tiny$\mbox{}+ 4{x}^{6}{y}^{3}
  $}}  \put(125,60){    \makebox[3.0 \unitlength]{\tiny$\mbox{}+ 2{x}^{5}{y}^{4}
  $}}  \put(75,40){    \makebox[3.0 \unitlength]{\tiny$\mbox{}-4{x}^{3}{y}^{6}
  $}}  \put(50,30){    \makebox[3.0 \unitlength]{\tiny$\mbox{}-4{x}^{2}{y}^{7}
  $}}  \put(175,90){    \makebox[3.0 \unitlength]{\tiny$\mbox{}+ \phantom{4}{x}^{7}y
  $}}  \put(125,70){    \makebox[3.0 \unitlength]{\tiny$\mbox{}+ 3{x}^{5}{y}^{3}
  $}}  \put(100,60){    \makebox[3.0 \unitlength]{\tiny$\mbox{}\color{purple}+ \phantom{4}{x}^{4}{y}^{4}
  $}}  \put(75,50){    \makebox[3.0 \unitlength]{\tiny$\mbox{}-2{x}^{3}{y}^{5}
  $}}  \put(50,40){    \makebox[3.0 \unitlength]{\tiny$\mbox{}-3{x}^{2}{y}^{6}
  $}}  \put(25,30){    \makebox[3.0 \unitlength]{\tiny$\mbox{}+ 5x{y}^{7}
  $}}  \put(150,90){    \makebox[3.0 \unitlength]{\tiny$\mbox{}+ \phantom{4}{x}^{6}y
  $}}  \put(100,70){    \makebox[3.0 \unitlength]{\tiny$\mbox{}+ 2{x}^{4}{y}^{3}
  $}}  \put(50,50){    \makebox[3.0 \unitlength]{\tiny$\mbox{}-3{x}^{2}{y}^{5}
  $}}  \put(25,40){    \makebox[3.0 \unitlength]{\tiny$\mbox{}+ 4x{y}^{6}
  $}}  \put(0,30){    \makebox[3.0 \unitlength]{\tiny$\mbox{}-3{y}^{7}
  $}}  \put(125,90){    \makebox[3.0 \unitlength]{\tiny$\mbox{}+ \phantom{4}{x}^{5}y
  $}}  \put(75,70){    \makebox[3.0 \unitlength]{\tiny$\mbox{}\color{purple}+ \phantom{4}{x}^{3}{y}^{3}
  $}}  \put(50,60){    \makebox[3.0 \unitlength]{\tiny$\mbox{}- \phantom{4}{x}^{2}{y}^{4}
  $}}  \put(0,40){    \makebox[3.0 \unitlength]{\tiny$\mbox{}- \phantom{4}{y}^{6}
  $}}  \put(100,90){    \makebox[3.0 \unitlength]{\tiny$\mbox{}+ \phantom{4}{x}^{4}y
  $}}  \put(0,50){    \makebox[3.0 \unitlength]{\tiny$\mbox{}-3{y}^{5}
  $}}  \put(75,90){    \makebox[3.0 \unitlength]{\tiny$\mbox{}+ \phantom{4}{x}^{3}y
  $}}  \put(25,70){    \makebox[3.0 \unitlength]{\tiny$\mbox{}- \phantom{4}x{y}^{3}
  $}}  \put(0,60){    \makebox[3.0 \unitlength]{\tiny$\mbox{}- \phantom{4}{y}^{4}
  $}}  \put(50,90){    \makebox[3.0 \unitlength]{\tiny$\mbox{}+ \phantom{4}{x}^{2}y
  $}}  \put(0,70){    \makebox[3.0 \unitlength]{\tiny$\mbox{}-3{y}^{3}
  $}}  \put(25,90){    \makebox[3.0 \unitlength]{\tiny$\mbox{}\color{purple}+ \phantom{4}xy
  $}}  \put(0,80){    \makebox[3.0 \unitlength]{\tiny$\mbox{}- \phantom{4}{y}^{2}
  $}}  \put(0,90){    \makebox[3.0 \unitlength]{\tiny$\mbox{}+ \phantom{4}y
  $}}
  \put(175,20){\tiny$\mbox{} + \dotsb$}
  \end{picture}
\end{center}
 \vspace*{-3ex}
\caption{\label{IssacSubmission:fig:Diagonal}A diagonal of a bivariate rational series
(see Ex.~\ref{IssacSubmission:example:AQuarticEquation}).}
\end{figure}
\end{example}
We will follow the path drawn in~\cite{Christol74}, and compute the $N$th coefficient~$f_N$ using $a/b$ as intermediate data structure.

\subsection{Nth coefficient via diagonals}\label{IssacSubmission:sec:SecOp&Diag}

Bivariate section operators can be defined by mimicking~\eqref{IssacSubmission:eq:SectionOperatorsDefinition}. Although they  depend on two residues modulo~$p$, we will always use them with the same residue for both variables; thus we will simply denote them by~$S_r$ for $0 \leq r < p$. A nice feature is that they commute  with the diagonal operator,
\begin{equation*}\label{IssacSubmission:eq:CommutationRuleDiagonalSections}
  S_r D = D S_r.
\end{equation*}
Together with~\eqref{IssacSubmission:eq:RationalFunctionDiagonal}, this property translates the action of the section operators from the algebraic series~$f$ to the rational function~$a/b$. Moreover, property~\eqref{IssacSubmission:eq:SectionOperatorsAndProduct} extends to bivariate section operators and  justifies the following computation
\begin{multline*}\label{IssacSubmission:eq:FrobeniusTrick}
  S_r \frac{a(x,y)}{b(x,y)} = S_r \frac{a(x,y) b(x,y)^{p-1}}{b(x,y)^p} \\ = 
  S_r \frac{a(x,y) b(x,y)^{p-1}}{b(x^p,y^p)} = \frac{S_r a(x,y) b(x,y)^{p-1}}{b(x,y)}.
\end{multline*}
To put it plainly
\begin{equation*}
  S_r \frac{a}{b} = \frac{S_r a b^{p-1}}{b}.
\end{equation*}
This yields an action on bivariate polynomials~$v$ of what we may call pseudo-section operators,  defined by
\begin{equation*}
  T_r v = S_r v b^{p-1},\qquad 0 \leq r < p,
\end{equation*}
hence $$(S_{r_1} \cdots S_{r_k})\left(\frac{a}{b}\right) = \frac{(T_{r_1} \cdots T_{r_k})(a)}{b},\qquad 0 \leq r_i < p.$$
At this point, it is not difficult to see that the vector subspace~${\bF_p}[x,y]_{d_x,d_y}$, where~$d_x$ and~$d_y$ are the maximal partial degrees of~$a$ and~$b$ defined in~\eqref{IssacSubmission:eq:MaxPartialDegreesab}, is stable under the pseudo-section operators. The supports of~$b$  and of polynomials~$v$ in~${\bF_p}[x,y]_{d_x,d_y}$ belong to a small rectangle $[0,d_x] \times [0,d_y]$. 
Raising~$b$ to the power~$p-1$ enlarges the rectangle by a factor~$p-1$ and the product by~$v$ results in a rectangle $p$ times larger than the small rectangle. The pseudo-section operators, essentially based on a Euclidean division by~$p$, send the large rectangle back into the small rectangle. 

As a consequence we have at our disposal a new linear representation $\big( L$, $(A_r)_{0 \leq r < p}, C\big)$, this time based on the vector space~${\bF_p}[x,y]_{d_x,d_y}$, and on its canonical basis. The row vector~$L$ has all its components equal to~$0$ except the first one, equal to~$1$, because all monomials of the canonical basis take the value~$0$ at $x = 0$, $y = 0$, except the monomial~$1$. The family~$(A_r)_{0 \leq r < p}$ of matrices of size $(1+d_x)(1+d_y)$ expresses the pseudo-section operators~$T_r$. 
The column vector~$C$ contains
the coordinates of the polynomial~$a$. 

The crucial difference with the linear representation in Sec.~\ref{IssacSubmission:sec:NthCoefficientViaDACRecurrence} is that the dimension of the new one is much smaller.

It remains to compute~$f_N$ by using the rational nature with respect to the radix~$p$ of the power series~$f(x)$. The process is summarized in Algorithm~\ref{IssacSubmission:algo:Christol}. Anew we conclude that the arithmetic complexity for the computing of~$f_N$ is~$\bigO(\log N)$. More precisely, since the linear representation has size $(1+d_x)(1+d_y)$,
the computation of~$f_N$ is of order~$\bigO(d_x^2 d_y^2 \log N )$. 

\begin{algo}	
	Algorithm \textbf{Nth coefficient via diagonals}
	
		\begin{algoenv}{A polynomial $E(x,y)\in \bF_p[x,y]$ with 
	$E(0,0) = 0$ and $E_y(0,0) \neq 0$, and an integer~$N\geq 0$.}
	{The $N$th coefficient~$f_N$ of the unique formal power series~$f\in\bF_p[[x]]$ solution of $E(x,f(x)) = 0$, $f(0) = 0$.}
    \State1. Compute $a(x,y) = y E_y(xy,y), b(x,y) = E(xy,y)/y$.
	\State \Comment At this point $f = D(a/b)$.

    \State2. Set $d_x \!=\! \max(\deg_x a, \deg_x b)$,  $d_y \!=\! \max(\deg_y a, \deg_y b)$, $L = [1, 0, \cdots, 0]$, and $C$ the column vector expressing $a(x,y)$ in the canonical basis of~${\bF_p}[x,y]_{d_x,d_y}$. 

    \State3. Compute $B = b^{p-1}$.

    \State 4. {\textbf{for}}  {$r$ from~$0$ to~$p-1$}
      \State \rule{2.3em}{0ex}  compute the image of the canonical basis\\ 
           \rule{2.3em}{0ex} $T_r x^i y^j = \sum_{n,m} A_{r; i,j; n,m} x^n y^m$\\
           \rule{2.3em}{0ex} and store the result in the matrix~$A_r$

    \State5. \Return $f_{N} = {L A_{N_{\ell}}\dotsb A_{N_0}\frac{C}{b(0,0)}}$ for $N = (N_{\ell} \ldots N_0)_p$.
	\end{algoenv}
	\caption{\label{IssacSubmission:algo:Christol}$N$th coefficient, via diagonals.}
\end{algo}

\begin{example}[A cubic equation, cont.]\label{IssacSubmission:example:ACubicEquationContinued} 
We return to Ex.~\ref{IssacSubmission:example:ACubicEquation} with  $p = 7$. We find $a = - y - x y^2 + 3 y^3+ 3 x y^4 + 2 x^2 y^4$ and $b = -1 + x - x y + y^2 + x y^3 + x^2 y^3$, hence $d_x = 2$, $d_y = 4$. Next we compute $B = b^{p-1}$ and we build the linear representation. 
The main point is the size of the representation, namely $(1 + d_x)(1 + d_y) = 15$. This is ridiculously small as compared to the value obtained in Ex.~\ref{IssacSubmission:example:ACubicEquation}. Even more importantly, the size remains the same if we change the prime number. 

Although it is not necessary that we see the linear representation with our human eyes, it is quite normal that we watch it if only to control the computations. The matrices~$A_r$, whose size is $(1+d_x)(1+d_y)$, have indices which are pairs of pairs. We can view them as square matrices of size~$1+d_x$ whose elements are square blocks of size~$1+d_y$. The coordinates over~$x^i y^j$ of the image~$T_r x^n y^m$ is in the block with indices~$(i,n)$ at place~$(j,m)$.

\begin{figure}[t]
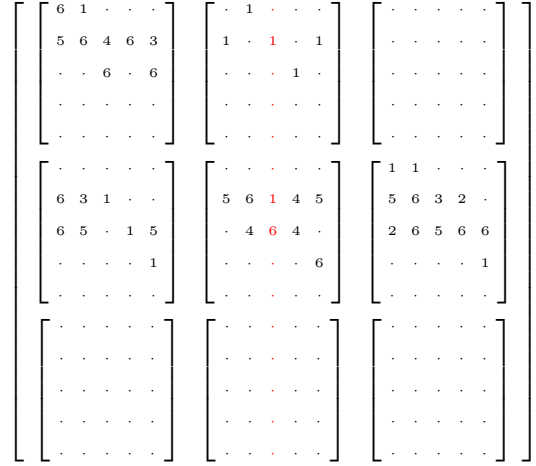

{\tiny
\[\arraycolsep=2.7pt
\left[ \begin {array}{ccc}  \left[ \begin {array}{ccccc} 6&1&\makebox[1.5ex]{$\cdot$}&\makebox[1.5ex]{$\cdot$}&\makebox[1.5ex]{$\cdot$}
\\ \noalign{\medskip}5&6&4&6&3\\ \noalign{\medskip}\makebox[1.5ex]{$\cdot$}&\makebox[1.5ex]{$\cdot$}&6&\makebox[1.5ex]{$\cdot$}&6
\\ \noalign{\medskip}\makebox[1.5ex]{$\cdot$}&\makebox[1.5ex]{$\cdot$}&\makebox[1.5ex]{$\cdot$}&\makebox[1.5ex]{$\cdot$}&\makebox[1.5ex]{$\cdot$}\\ \noalign{\medskip}\makebox[1.5ex]{$\cdot$}&\makebox[1.5ex]{$\cdot$}&\makebox[1.5ex]{$\cdot$}&\makebox[1.5ex]{$\cdot$}&\makebox[1.5ex]{$\cdot$}
\end {array} \right] & \left[ \begin {array}{ccccc} \makebox[1.5ex]{$\cdot$}&1&{\color{red}\makebox[1.5ex]{$\cdot$}}&\makebox[1.5ex]{$\cdot$}&\makebox[1.5ex]{$\cdot$}
\\ \noalign{\medskip}1&\makebox[1.5ex]{$\cdot$}&{\color{red}1}&\makebox[1.5ex]{$\cdot$}&1\\ \noalign{\medskip}\makebox[1.5ex]{$\cdot$}&\makebox[1.5ex]{$\cdot$}&{\color{red}\makebox[1.5ex]{$\cdot$}}&1&\makebox[1.5ex]{$\cdot$}
\\ \noalign{\medskip}\makebox[1.5ex]{$\cdot$}&\makebox[1.5ex]{$\cdot$}&{\color{red}\makebox[1.5ex]{$\cdot$}}&\makebox[1.5ex]{$\cdot$}&\makebox[1.5ex]{$\cdot$}\\ \noalign{\medskip}\makebox[1.5ex]{$\cdot$}&\makebox[1.5ex]{$\cdot$}&{\color{red}\makebox[1.5ex]{$\cdot$}}&\makebox[1.5ex]{$\cdot$}&\makebox[1.5ex]{$\cdot$}
\end {array} \right] & \left[ \begin {array}{ccccc} \makebox[1.5ex]{$\cdot$}&\makebox[1.5ex]{$\cdot$}&\makebox[1.5ex]{$\cdot$}&\makebox[1.5ex]{$\cdot$}&\makebox[1.5ex]{$\cdot$}
\\ \noalign{\medskip}\makebox[1.5ex]{$\cdot$}&\makebox[1.5ex]{$\cdot$}&\makebox[1.5ex]{$\cdot$}&\makebox[1.5ex]{$\cdot$}&\makebox[1.5ex]{$\cdot$}\\ \noalign{\medskip}\makebox[1.5ex]{$\cdot$}&\makebox[1.5ex]{$\cdot$}&\makebox[1.5ex]{$\cdot$}&\makebox[1.5ex]{$\cdot$}&\makebox[1.5ex]{$\cdot$}
\\ \noalign{\medskip}\makebox[1.5ex]{$\cdot$}&\makebox[1.5ex]{$\cdot$}&\makebox[1.5ex]{$\cdot$}&\makebox[1.5ex]{$\cdot$}&\makebox[1.5ex]{$\cdot$}\\ \noalign{\medskip}\makebox[1.5ex]{$\cdot$}&\makebox[1.5ex]{$\cdot$}&\makebox[1.5ex]{$\cdot$}&\makebox[1.5ex]{$\cdot$}&\makebox[1.5ex]{$\cdot$}
\end {array} \right] \\ \noalign{\medskip} \left[ \begin {array}
{ccccc} \makebox[1.5ex]{$\cdot$}&\makebox[1.5ex]{$\cdot$}&\makebox[1.5ex]{$\cdot$}&\makebox[1.5ex]{$\cdot$}&\makebox[1.5ex]{$\cdot$}\\ \noalign{\medskip}6&3&1&\makebox[1.5ex]{$\cdot$}&\makebox[1.5ex]{$\cdot$}\\ \noalign{\medskip}6&
5&\makebox[1.5ex]{$\cdot$}&1&5\\ \noalign{\medskip}\makebox[1.5ex]{$\cdot$}&\makebox[1.5ex]{$\cdot$}&\makebox[1.5ex]{$\cdot$}&\makebox[1.5ex]{$\cdot$}&1\\ \noalign{\medskip}\makebox[1.5ex]{$\cdot$}&\makebox[1.5ex]{$\cdot$}&\makebox[1.5ex]{$\cdot$}&\makebox[1.5ex]{$\cdot$}&\makebox[1.5ex]{$\cdot$}
\end {array} \right] & \left[ \begin {array}{ccccc} \makebox[1.5ex]{$\cdot$}&\makebox[1.5ex]{$\cdot$}&{\color{red}\makebox[1.5ex]{$\cdot$}}&\makebox[1.5ex]{$\cdot$}&\makebox[1.5ex]{$\cdot$}
\\ \noalign{\medskip}5&6&{\color{red}1}&4&5\\ \noalign{\medskip}\makebox[1.5ex]{$\cdot$}&4&{\color{red}6}&4&\makebox[1.5ex]{$\cdot$}
\\ \noalign{\medskip}\makebox[1.5ex]{$\cdot$}&\makebox[1.5ex]{$\cdot$}&{\color{red}\makebox[1.5ex]{$\cdot$}}&\makebox[1.5ex]{$\cdot$}&6\\ \noalign{\medskip}\makebox[1.5ex]{$\cdot$}&\makebox[1.5ex]{$\cdot$}&{\color{red}\makebox[1.5ex]{$\cdot$}}&\makebox[1.5ex]{$\cdot$}&\makebox[1.5ex]{$\cdot$}
\end {array} \right] & \left[ \begin {array}{ccccc} 1&1&\makebox[1.5ex]{$\cdot$}&\makebox[1.5ex]{$\cdot$}&\makebox[1.5ex]{$\cdot$}
\\ \noalign{\medskip}5&6&3&2&\makebox[1.5ex]{$\cdot$}\\ \noalign{\medskip}2&6&5&6&6
\\ \noalign{\medskip}\makebox[1.5ex]{$\cdot$}&\makebox[1.5ex]{$\cdot$}&\makebox[1.5ex]{$\cdot$}&\makebox[1.5ex]{$\cdot$}&1\\ \noalign{\medskip}\makebox[1.5ex]{$\cdot$}&\makebox[1.5ex]{$\cdot$}&\makebox[1.5ex]{$\cdot$}&\makebox[1.5ex]{$\cdot$}&\makebox[1.5ex]{$\cdot$}
\end {array} \right] \\ \noalign{\medskip} \left[ \begin {array}
{ccccc} \makebox[1.5ex]{$\cdot$}&\makebox[1.5ex]{$\cdot$}&\makebox[1.5ex]{$\cdot$}&\makebox[1.5ex]{$\cdot$}&\makebox[1.5ex]{$\cdot$}\\ \noalign{\medskip}\makebox[1.5ex]{$\cdot$}&\makebox[1.5ex]{$\cdot$}&\makebox[1.5ex]{$\cdot$}&\makebox[1.5ex]{$\cdot$}&\makebox[1.5ex]{$\cdot$}\\ \noalign{\medskip}\makebox[1.5ex]{$\cdot$}&\makebox[1.5ex]{$\cdot$}
&\makebox[1.5ex]{$\cdot$}&\makebox[1.5ex]{$\cdot$}&\makebox[1.5ex]{$\cdot$}\\ \noalign{\medskip}\makebox[1.5ex]{$\cdot$}&\makebox[1.5ex]{$\cdot$}&\makebox[1.5ex]{$\cdot$}&\makebox[1.5ex]{$\cdot$}&\makebox[1.5ex]{$\cdot$}\\ \noalign{\medskip}\makebox[1.5ex]{$\cdot$}&\makebox[1.5ex]{$\cdot$}&\makebox[1.5ex]{$\cdot$}&\makebox[1.5ex]{$\cdot$}&\makebox[1.5ex]{$\cdot$}
\end {array} \right] & \left[ \begin {array}{ccccc} \makebox[1.5ex]{$\cdot$}&\makebox[1.5ex]{$\cdot$}&{\color{red}\makebox[1.5ex]{$\cdot$}}&\makebox[1.5ex]{$\cdot$}&\makebox[1.5ex]{$\cdot$}
\\ \noalign{\medskip}\makebox[1.5ex]{$\cdot$}&\makebox[1.5ex]{$\cdot$}&{\color{red}\makebox[1.5ex]{$\cdot$}}&\makebox[1.5ex]{$\cdot$}&\makebox[1.5ex]{$\cdot$}\\ \noalign{\medskip}\makebox[1.5ex]{$\cdot$}&\makebox[1.5ex]{$\cdot$}&{\color{red}\makebox[1.5ex]{$\cdot$}}&\makebox[1.5ex]{$\cdot$}&\makebox[1.5ex]{$\cdot$}
\\ \noalign{\medskip}\makebox[1.5ex]{$\cdot$}&\makebox[1.5ex]{$\cdot$}&{\color{red}\makebox[1.5ex]{$\cdot$}}&\makebox[1.5ex]{$\cdot$}&\makebox[1.5ex]{$\cdot$}\\ \noalign{\medskip}\makebox[1.5ex]{$\cdot$}&\makebox[1.5ex]{$\cdot$}&{\color{red}\makebox[1.5ex]{$\cdot$}}&\makebox[1.5ex]{$\cdot$}&\makebox[1.5ex]{$\cdot$}
\end {array} \right] & \left[ \begin {array}{ccccc} \makebox[1.5ex]{$\cdot$}&\makebox[1.5ex]{$\cdot$}&\makebox[1.5ex]{$\cdot$}&\makebox[1.5ex]{$\cdot$}&\makebox[1.5ex]{$\cdot$}
\\ \noalign{\medskip}\makebox[1.5ex]{$\cdot$}&\makebox[1.5ex]{$\cdot$}&\makebox[1.5ex]{$\cdot$}&\makebox[1.5ex]{$\cdot$}&\makebox[1.5ex]{$\cdot$}\\ \noalign{\medskip}\makebox[1.5ex]{$\cdot$}&\makebox[1.5ex]{$\cdot$}&\makebox[1.5ex]{$\cdot$}&\makebox[1.5ex]{$\cdot$}&\makebox[1.5ex]{$\cdot$}
\\ \noalign{\medskip}\makebox[1.5ex]{$\cdot$}&\makebox[1.5ex]{$\cdot$}&\makebox[1.5ex]{$\cdot$}&\makebox[1.5ex]{$\cdot$}&\makebox[1.5ex]{$\cdot$}\\ \noalign{\medskip}\makebox[1.5ex]{$\cdot$}&\makebox[1.5ex]{$\cdot$}&\makebox[1.5ex]{$\cdot$}&\makebox[1.5ex]{$\cdot$}&\makebox[1.5ex]{$\cdot$}
\end {array} \right] \end {array} \right]
\]
}
\vspace*{-1ex}
\caption{\label{IssacSubmission:fig:ACubicEquationMatrixA_1}The matrix~$A_1$ in Ex.~\ref{IssacSubmission:example:ACubicEquationContinued}.}
\end{figure}

We only display the matrix~$A_1$ (the zero coefficients are replaced by dots), in Fig.~\ref{IssacSubmission:fig:ACubicEquationMatrixA_1}. It gives us for example the value of~$T_1 xy^2$. Because the exponent of~$x$ is~$1$, we look at the column of blocks number~$1$ (the second one) and because the exponent of~$y$ is~$2$ we look at the column number~$2$ (the third one) in the matrices of this column (red column). In the first block, which corresponds to~$x^0$, we see a coefficient~$1$ in row number~$1$, hence a term~$x^0 y^1 = y$. In the second block, we see a coefficient~$1$ in row number~$1$ hence a term $x^1 y^1 = xy$ and a coefficient~$6$ in row number~$2$ hence a term $6 x^1 y^2 = 6 xy^2$. We conclude $T_1 xy^2 = y + xy + 6xy^2$.

\end{example}

\subsection{Precomputation}\label{IssacSubmission:subsec:SecOp&Diag-Precomputation}

All the needed information to build the linear representation is contained in the polynomial $B = b^{p-1}$. Indeed, when we compute the image of an element~$x^i y^j$ of the canonical basis by the pseudo-section operator~$T_r$, we first multiply~$B$ by~$x^i y^j$; this transformation is merely a translation of the exponents, that requires no arithmetic operation. Next, we extract the monomials whose exponents are congruent to~$r$ modulo~$p$, again without any computation in $\bF_p$. The conclusion is that, apart the final computation to obtain the value~$f_N$, the cost comes from the computation of $B=b^{p-1}$. This can be done using binary powering and Kronecker's substitution. Since~$B$ has partial degrees at most~$p d_x$ in $x$ and at most~$p d_y$ in $y$, the arithmetic complexity is~$\softO(p^2 d_x d_y)$. 

Gathering the study of the computation and the precomputation, and using Eq.~\eqref{IssacSubmission:eq:MaxPartialDegreesab}, we obtain the following result.
\begin{thm}\label{IssacSubmission:thm:ChristolComplexity}
Let $E$ be in $\bF_p[x,y]_{h,d}$ such that $E(0,0) =
0$ and $E_y(0,0) \neq 0$, and let $f\in\bF_p[[x]]$ be its unique root with $f(0)=0$.
Algorithm~\ref{IssacSubmission:algo:Christol} computes the $N$th coefficient~$f_N$ of~$f$
in
$O(h^2 (d+h)^2 \log N) + \softO(p^2 h(d+h))$
 operations in~$\bF_p$.
\end{thm}

The striking points are the decreasing of the exponent of~$p$ from~$3d$ to~$2$, and the replacement of the \emph{multiplicative} complexity {$p^{3d} \times \log N$} of Theorem~\ref{IssacSubmission:thm:CKMRComplexity} by the \emph{additive} complexity $p^2 + \log N$.
This is already a tremendous improvement. 
In the next section, we will present a further improvement.

\section{A Faster Algorithm}\label{sec:newalgo}
We will improve the algorithm of the previous section in order to decrease the exponent of the prime~$p$ from~$2$ to~$1$.

\subsection{A small part of the large power is enough}\label{ssec:SmallPart}
As pointed out in Sec.~\ref{IssacSubmission:subsec:SecOp&Diag-Precomputation}, all the information needed to build the square matrices $A_0, \ldots, A_{p-1}$ of the linear representation 
is contained in the large power~$B = b^{p-1}$. However building the $A_i$'s does not require the whole polynomial~$B$. 

To see this, assume $0 \leq \alpha \leq (p-1) d_x$ and $0 \leq \beta \leq (p-1) d_y$
are such that the monomial $x^\alpha y^\beta$ of $B$ contributes to one of the matrices $A_r$. Since $A_{r;i,j;n,m}$ is the coefficient of $x^n y^m$ in
$T_r (x^i y^j) = S_r (x^i y^j B)$, this means that $i+\alpha  \equiv j + \beta \bmod p$ for some $0\leq i \leq d_x$ and $0\leq j \leq d_y$.
As a consequence, the difference exponent $\delta = \beta - \alpha$ necessarily belongs to $[-d_y,d_x] + p\bZ$.
For large~$p$, this means that only a fraction~$1/p$ of~$B$ is useful. The aim of this section is to take an algorithmic advantage of this remark. 

\begin{example}[A cubic equation, cont.]\label{IssacSubmission:example:ACubicEquationContinuedContinued} 
The previous 
 phenomenon is exemplified in Fig.~\ref{IssacSubmission:fig:Strips} for the equation $E(x,y) = x - (1+x)y + x^2 y^2 + (1+x) y^3 = 0$,  already encountered in Ex.~\ref{IssacSubmission:example:ACubicEquation} and~\ref{IssacSubmission:example:ACubicEquationContinued}. We use $p=11$ on the left-hand side, and $p=109$ on the right-hand side. The polynomial~$B$ is represented by its monomial support (in black).
Besides, the points~$(\alpha,\beta)$ defined by $\alpha = p k + r - i$, $\beta  = p \ell + r - j$, with $0 \leq i \leq d_x$, $0 \leq j \leq d_y$ are colored, with a color which goes from blue to red when~$r$ goes from 0 to~$p-1$ (but the pieces overlap). This generates strips with slope~$1$: only monomials in these strips contribute to the linear representation. For~$p$ small ($p = 11$), the strips cover almost the whole rectangle, but for~$p$ moderately large ($p = 109$) the proportion of useful coefficients in~$B$ becomes much  smaller.

\begin{figure}[t]
  \begin{center}
    \includegraphics[width=0.4\linewidth]{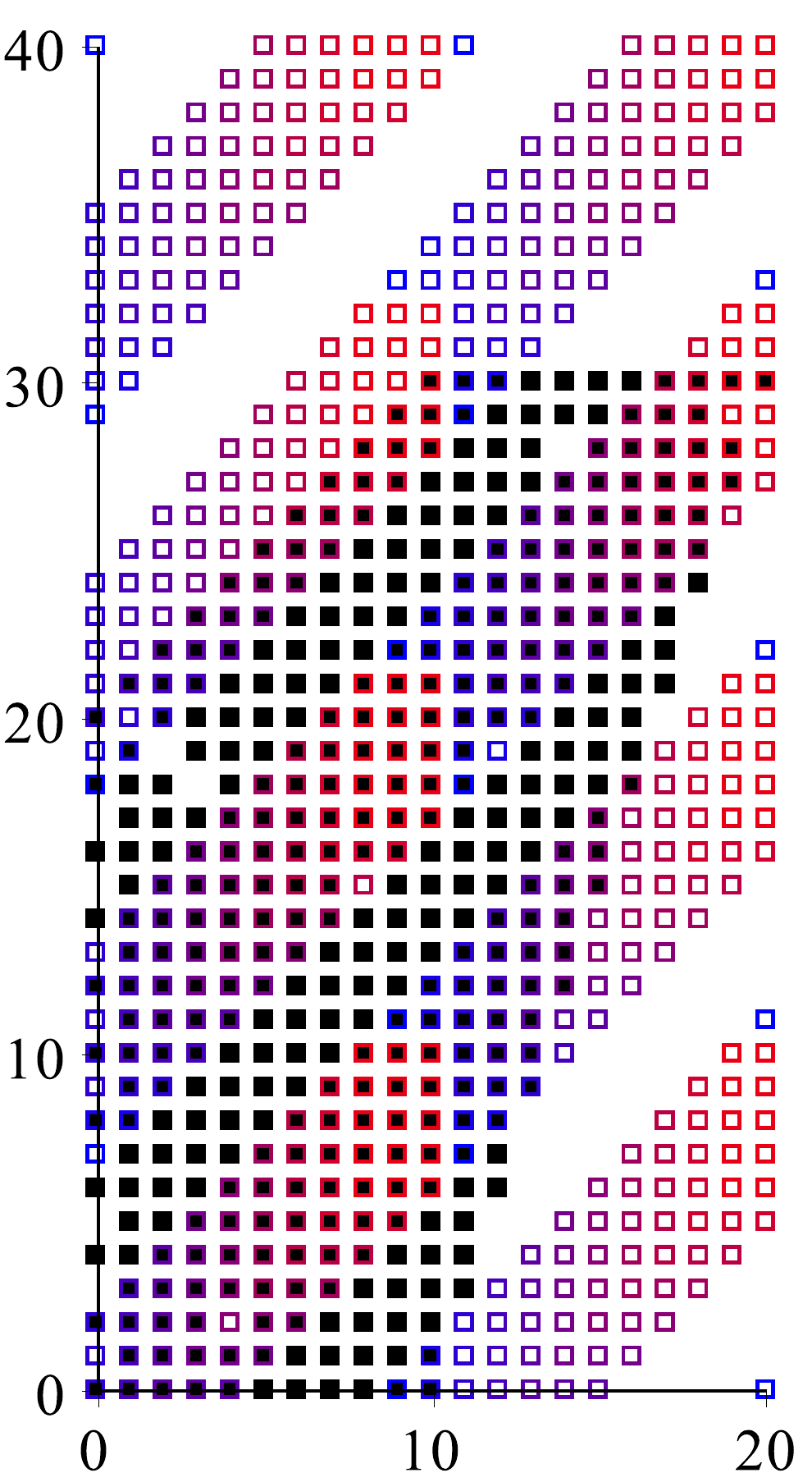}
    \hfil
    \includegraphics[width=0.4\linewidth]{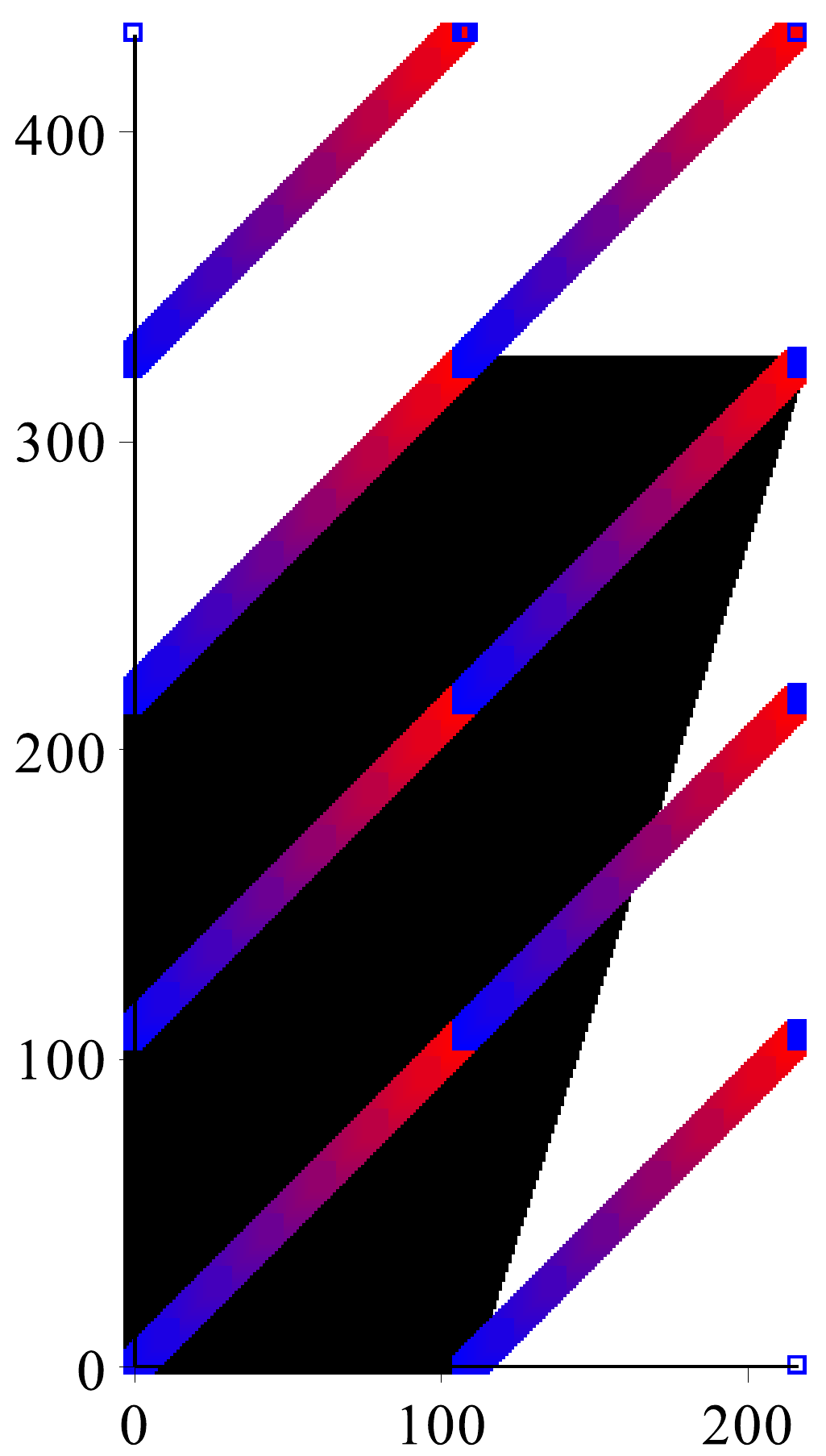}
  \end{center}
  \vspace*{-3ex}
  \caption{\label{IssacSubmission:fig:Strips}The useful coefficients of~$B = b^{p-1}$ are located in the diagonal strips. See Ex.~\ref{IssacSubmission:example:ACubicEquationContinuedContinued} for explanations.}
\end{figure}

\end{example}

To emphasize the difference exponent $\delta = \beta - \alpha$,
we perform a change of variables $x := x/t$, $y := t$, which produces Laurent polynomials with respect to~$t$, namely
$
b(x/t,t)$ and 
\begin{equation*}\label{IssacSubmission:eq:Definition-pi_delta(x)}
B(x/t,t)
=
b(x/t,t)^{p-1}
= \sum_{\delta} \pi_\delta(x) t^\delta.
\end{equation*}
By construction, the coefficient $\pi_\delta(x)$ is a polynomial, namely the $\delta$-diagonal $\sum_{i} B_{i,i+\delta}x^i$ of $B(x,y)=\sum_{i,j} B_{i,j}x^i y^j$.

Let~$\delta_-$ be the opposite of the lowest degree with respect to~$t$ of~$b(x/t,t)$, and~$\delta_+$ be the degree with respect to~$t$ of~$b(x/t,t)$.
An immediate bound is $\delta_- \leq d_x$ and $\delta_+ \leq d_y$.

By the previous discussion, all the information needed to build up the matrices of the linear representation is contained in the polynomials $\pi_\delta(x)$, for $\delta$ in \[\Delta := ([-d_y,d_x] + p\mathbb{Z}) \cap [(1-p)\delta_-, (p-1)\delta_+].\]
This set is 
a union of small intervals of length $d_x + d_y +1$ regularly spaced by~$p$ and within~$[(1-p)d_x, (p-1)d_y]$. 

\subsection{Partial bivariate powering}\label{ssec:PartialPowering}
Our aim is to show that the collection of $\pi_\delta(x)$ for $\delta\in\Delta$ can be computed in complexity quasi-linear in~$p$. 
The starting point is to write
\begin{equation*}\label{IssacSubmission:eq:Reversing}
  B(x/t,t) = \frac{b(x/t,t)^p}{b(x/t,t)} = \frac{b(x^p/t^p,t^p)}{b(x/t,t)}.
\end{equation*}
Here~$b(x^p/t^p,t^p)$ is obtained by a mere rewriting of~$b(x/t,t)$, hence at no cost, while $1/b(x/t,t)$ is a rational function.
To be more concrete, we denote 
\begin{equation*}
  b(x/t,t) = \sum_{v=-\delta_-}^{\delta_+} b_v(x) t^{v},
  \qquad
  \frac{1}{b(x/t,t)} = \sum_{u \geq \delta_-} c_u(x) t^u,
\end{equation*}
where the coefficients $b_v(x)$ are polynomials in $x$ and 
the coefficients $c_v(x)$ are rational functions in $x$.

Now, the polynomial~$\pi_\delta(x)$ is expressed by a convolution 
\begin{equation}\label{IssacSubmission:eq:Convolution}
  \pi_\delta(x) = \sum_{u + p v = \delta} c_u(x) b_v(x^p).
\end{equation}
The constraint $u + p v = \delta$ implies $u \equiv \delta \bmod p$.
Therefore, it is sufficient to compute the $c_\delta(x)$ for all $\delta \in \Delta' = p \delta_- + \Delta$.

To do this efficiently, we will exploit the fact that the sequence of rational functions $(c_u(x))_u$ satisfies a linear recurrence with coefficients in ${\bF_p}[x]$. The coefficients of the recurrence relation are those of the polynomial $t^{\delta_-} b(x/t,t)$.

We are facing the following issue: since we need the values of $c_\delta(x)$ for indices of order $\delta=\Theta(p)$, we can not unroll the recurrence relation directly, otherwise the complexity (and actually the mere arithmetic size of the rational functions $c_i(x), i < u$, computed along the way) becomes quadratic in~$p$.
Fortunately, one can compute the $N$th term of a recurrent sequence without computing all the previous coefficients. The key is the following lemma.

\begin{lem}\emph{\cite{Fiduccia85}}
	\label{IssacSubmission::lem:Fiduccia}
Let~$R$ be a commutative ring and~$a(t)/b(t)$ be a rational function in~$R(t)$ with~$b(0)$ invertible in $R$. Let~$d$ be the degree of~$b(t)$ and~$b^*(t) = t^d b(1/t)$ be its reciprocal. For any integer~$N\geq 0$, the coefficients~$c_N$, $c_{N+1}$, $\ldots$, $c_{N+d-1}$ of the formal power series expansion of $a(t)/b(t) = \sum_{n\geq 0} c_n t^n$ are the coefficients of~$t^{2d-2}$, $\ldots$, $t^d$, $t^{d-1}$ in the product
\begin{equation*}
	(t^N \bmod b^*(x)) \times (c_{2d-2} + \dotsb + c_0 t^{2d-2}).
\end{equation*}
\end{lem}
Fiduccia's algorithm~\cite{Fiduccia85} relies on Lemma~\ref{IssacSubmission::lem:Fiduccia} and computes $t^N \bmod b^*(t)$ by binary powering. 
As a consequence, the $N$th term of a linear recurrent sequence of order $d$ is computed in $\softO(d \log N)$ operations in $R$. In our setting, we use Fiduccia's algorithm for $R=\bF_p(x)$. 

\begin{thm}\label{thm:PartialPowering}
Let $b$ be in $\bF_p[x,y]_{d_x, d_y}$ and let $\delta \in \mathbb{Z}$.
Algorithm \ref{algo:PartialPowering} computes the $\delta$-diagonal $\pi_\delta(x) = \sum_i B_{i,i+\delta} x^i$ of $B=b^{p-1}$
in $\softO(d_x(d_x+d_y)^3 p)$ operations in $\bF_p$.
\end{thm}

\begin{proof}
By~\eqref{IssacSubmission:eq:Convolution}, to compute $\pi_\delta$ it is sufficient to have the coefficients $c_{\delta-pv}$ for $-d_x \leq v \leq d_y$. The coefficient $c_N$ has the form
$p(x)/q(x)^{N+1}$, where $\deg p = O(N d_x)$, $\deg q \leq d_x$. 
The most expensive step in Fiduccia's algorithm for $c_N$ is computing $t^N$ modulo a polynomial in $\bF_p[x,y]_{d_x, d_x+d_y}$.
By Lemma~\ref{lem:boundsMahler}, this can be done in time $T_N = \softO(N d_x (d_x+d_y))$. Summing the costs $T_N$ over all $N$'s of the form ${\delta-pv}$ with $-d_x \leq v \leq d_y$ concludes the proof.
\end{proof}

\begin{algo}[ht]	
	Algorithm \textbf{PartialPowering}
	
	\begin{algoenv}{A polynomial $b$ in $\bF_p[x,y]$ and $\delta \in \mathbb{Z}$.}
{The $\delta$-diagonal of $B=b^{p-1}$, that is $\sum_i B_{i,i+\delta} x^i$.}

      \For{$v$ from $\textrm{val}_t \, b(x/t,t)$ to $\deg_t b(x/t,t)$}
        \State Compute $c_{\delta-pv}(x) = [t^{\delta-pv}] \frac{1}{b(x/t,t)}$ via Fiduccia's algo
	\EndFor
  \Return $\pi_\delta(x) = \sum_{u + p v = \delta} c_u(x) b_v(x^p)$
	\end{algoenv}
	\caption{\label{algo:PartialPowering}$\delta$-diagonal of a bivariate polynomial power.}
\end{algo}

\subsection{Faster precomputation} \label{ssec:FasterAlgo}
We finally combine the results in Sections~\ref{ssec:SmallPart} and~\ref{ssec:PartialPowering} with the ones of Sec.~\ref{sec:diag}, and modify Algorithm~\ref{IssacSubmission:algo:Christol} accordingly, by replacing its Step 3 with the partial powering routine (Algorithm~\ref{algo:PartialPowering}).
This yields our main result.

\begin{thm}\label{thm:faster}
Let $E$ be in $\bF_p[x,y]_{h,d}$ satisfy $E(0,0)=0$ and $E_y(0,0) \neq 0$, and let $f\in\bF_p[[x]]$ be its unique root with $f(0)=0$.
One can compute the  coefficient~$f_N$ of~$f$
in $O(h^2 (d+h)^2 \log N) + \softO(h(d+h)^5 p)$
 operations in~$\bF_p$.
\end{thm}

\begin{proof}
Only the cost of the precomputation changes.
Instead of computing the whole power $B=b^{p-1}$, we only compute 	
its $\delta$-diagonals $\pi_\delta(x)$, for $\delta \in \Delta'$.
Since $|\Delta'| = |\Delta| \leq (d_x+d_y+1)^2$, we conclude using Theorem~\ref{thm:PartialPowering}.
\end{proof}

\begin{example}[Catalan]
Let $E(x,y)=y-x-y^2$ be in $\bF_p[x,y]$, with root 
$f=\sum_{n \geq 0} C_n x^{n+1}$, where $C_n$ is 
$\frac{1}{n+1}\binom{2n}{n}$. Then $f$ is the diagonal of $a/b$ with $a=y(1-2y)$ and $b=1-x-y$. Thus $d_x=1, d_y=2$, and the linear representation has size~6. The needed information to compute the $6 \times 6$ matrices $A_0, \ldots, A_{p-1}$ consists in the $\delta$-diagonals of $B = b^{p-1}$ for $\delta \in \Delta := ([-2,1] + p\mathbb{Z}) \cap [1-p, p-1] = \{1-p,-2, -1, 0, 1, p-2, p-1 \}$. Writing $B(x/t,t) =(1-x/t-t)^{p-1} = \sum_{\delta=1-p}^{p-1} \pi_\delta(x)t^\delta$, the goal is to compute the polynomials $\pi_\delta(x)$ for $\delta\in\Delta$. This is done via the formula $B(x/t,t) = (1-x^p/t^p-t^p)/(1-x/t-t)$ which gives $\pi_\delta(x) = -c_{\delta+p}(x) x^p + c_\delta(x) - c_{\delta-p}(x)$. Therefore, it is sufficient to compute the rational functions $c_\delta(x)$ for all $\delta \in \Delta' = \Delta + p$. This is done using Fiduccia's algorithm, which exploits the fact that the $c_{\delta}$'s satisfy the recurrence $x c_{k+2}(x) = c_{k+1}(x)-c_k(x)$.

For instance, $\pi_0(x)=-x^p c_p(x)$, where $c_p(x)$ is obtained as the coefficient of $t^2$ in the product $(-1/x^2-t/x) \cdot (t^p \bmod -xt^2+t-1)$. This is done by binary powering. Checking that the algorithm correctly computes $\pi_0(x)$ amounts to checking the amusing identity:
\[\left[t^1\right] \; \left(t^p \bmod -xt^2+t-1 \right)=  \sum_{\ell=0}^{p-1} \binom{2\ell}{\ell} x^{\ell+1-p} \quad \text{in}\; \bF_p[x]. \]
\end{example}

\subsection{Practical issues and future work}
We have implemented the algorithms described in this section in \textsf{Maple}. Our prototype is available at the url 
	\newline\centerline{\url{http://specfun.inria.fr/dumas/Research/AlgModp/}.}
The implementation permits getting an idea of the feasibility
of the approach using diagonals. For instance, in the case of the quartic equation in Ex.~\ref{IssacSubmission:example:AQuarticEquation}, with $p=9001$ instead of~$p=11$, the algorithm spends $142$ seconds
on the precomputation, then computes the $N$-th coefficient for $N = 10^{10,000}$, resp. $10^{100,000}$, resp. $10^{1,000,000}$ in: $3$, resp. $25$, resp. $344$ seconds.
(Timings on a personal laptop, Intel Core i5, 2.8 GHz, 3MB.)

In a forthcoming article, we plan to extend our algorithms to all finite fields, and to improve the complexity results.

\begingroup
\let\oldbibliography\thebibliography
\renewcommand{\thebibliography}[1]{%
  \oldbibliography{#1}%
  \setlength{\itemsep}{1.5pt}%
}
\makeatletter
\def\section{%
    \@startsection{section}{1}{\z@}{-10\p@ \@plus -4\p@ \@minus -2\p@}
    {14\p@}{\baselineskip 14pt\secfnt\@ucheadtrue}%
}
\makeatother

\scriptsize
\smallskip
\bibliographystyle{abbrv}

\endgroup
\end{document}